\documentclass[submission]{eptcs}
\providecommand{\event}{QPL 2025} % Name of the event you are submitting to

\usepackage{iftex}

\usepackage{algorithm}
\usepackage{subcaption}
\usepackage{amsthm}
\usepackage{bm}
\usepackage{float}
\usepackage{tikz}
\usepackage{multirow}
\usepackage{hyperref}
\usepackage{enumitem}
\usepackage{amsmath}
\usepackage{amssymb}
\usepackage{babel}

\ifpdf
\else
\fi

\title{
  High-Precision Multi-Qubit Clifford+T\\Synthesis by Unitary Diagonalization
}

\author{Mathias Weiden \ \ \ Justin Kalloor \ \ \  John Kubiatowicz
    \institute{{University of California, Berkeley}}
    \email{
        {
            \{mtweiden, jkalloor3, kubitron\}@berkeley.edu %\quad jkalloor3@berkeley.edu \email{kubitron@cs.berkeley.edu}
        }
    }
\and
    Ed Younis \ \ \ Costin Iancu
    \institute{Lawrence Berkeley National Laboratory}
    \email{\{edyounis, cciancu\}@lbl.gov}
}
\def\titlerunning{High Precision Synthesis by Diagonalization}
\def\authorrunning{
    M. Weiden, et al.
}

\begin{document}
% Math
\newcommand{\unitarygroup}[0]{\mathbb{U}(2^n)}
\newcommand{\realnum}[0]{\mathbb{R}}
\newcommand{\complexnum}[0]{\mathbb{C}}
\newcommand{\phase}[1]{e^{i#1}}
\newcommand{\expectation}[0]{\mathbb{E}}
\newcommand{\integers}[0]{\mathbb{Z}}

% Synthesis
\newcommand{\target}[0]{U_{tar}}
\newcommand{\circuit}[0]{C_t}
\newcommand{\perturbation}[0]{U_{\epsilon}}

% MDPs
\newcommand{\initialstates}[0]{S_{I}}
\newcommand{\terminalstates}[0]{S_{T}}
\newcommand{\statespace}[0]{S}
\newcommand{\actionspace}[0]{A}
\newcommand{\s}[0]{s}
\newcommand{\mdpstate}[1]{\s_{#1}}
\newcommand{\ac}[0]{a}
\newcommand{\action}[1]{\ac_{#1}}
\newcommand{\trajectory}[0]{\tau}

\newcommand{\gridsynth}[0]{\text{\emph{gridsynth}}}

\long\def\comment#1{}
\long\def\note#1{{\em #1 }}
\def\parah#1{\vspace*{0.0in} \noindent{\bf #1:}}
\newcommand{\red}[1]{\textcolor{red}{#1}}

\maketitle
\newtheorem{theorem}{Theorem}

\begin{abstract}
Resource-efficient and high-precision approximate synthesis of quantum circuits expressed in the Clifford+T gate set is vital for Fault-Tolerant quantum computing.
Efficient optimal methods are known for single-qubit $R_Z$ unitaries, otherwise the problem is generally intractable. 
Search-based methods, like simulated annealing, empirically generate low resource cost approximate implement\-ations of general multi-qubit unitaries so long as low precision (Hilbert-Schmidt distances of $\epsilon \geq 10^{-2}$) can be tolerated.
These algorithms build up circuits that directly invert target unitaries.
We instead leverage search-based methods to first approximately diagonalize a unitary, then perform the inversion analytically.
This lets difficult continuous rotations be bypassed and handled in a post-processing step.
Our approach improves both the implementation precision and run time of synthesis algorithms by orders of magnitude when evaluated on unitaries from real quantum algorithms.
On benchmarks previously synthesizable only with analytical techniques like the Quantum Shannon Decomposition, diagonalization uses an average of 95\% fewer non-Clifford gates.

% {\it \footnotesize We present a method for synthesis of multi-qubit unitaries into
% circuits expressed in the Clifford+T Fault Tolerant gate set. The main
% idea is to pose synthesis as a matrix diagonalization problem instead
% of matrix inversion.  State-of-the-art multi-qubit FT
% synthesis algorithms have limitations: 1) they can generate only low
% precision solutions (e.g. $10^{-2}$); and 2) they cannot handle
% circuits with arbitrary rotations and therefore haven't been demonstrated on
% real programs. We enable a composable approach where other
% synthesis algorithms, such as optimal approximation of single-qubit
% $R_Z$ unitaries, are leveraged. We show how diagonalization can be easily retrofitted into other
% multi-qubit FT synthesis algorithms and enables orders of magnitude
% compilation time improvements and orders of magnitude solution
% precision improvements. We also show how multi-qubit diagonalization
% synthesis can be used in a generic compiler workflow.  So far, {\tt
% gridsynth} ($R_z$ synthesis), is the only practical solution available for FT
% circuit generation. When compared with {\tt gridsynth} based methods  that can
% handle ``arbitrary'' circuits, multi-qubit diagonalization based
% synthesis can reduce resource usage (T count or depth).  This work is a
% first demonstration of the practical value of multi-qubit FT synthesis
% algorithms and results indicate that there is
% value in augmenting compilers with diagonalization based synthesis and
% enabling a new class of algorithms.}

\end{abstract}

\section{Introduction}
\label{sec:intro}
Recent small-scale demonstrations of error-corrected quantum memory signal significant progress to\-ward the development of Fault-Tolerant (FT) quantum computers~\cite{acharya_2024_below_threshold, bluvstein_2023_ftquera, dasilva_2024_ftquantinuum}.
In this setting, programs must be expressed in universal FT gate sets, which usually consist of gates from the Clifford group and at least one non-Clifford gate for universality.
The Clifford+T gate set ($H, S, CNOT, T$) is one such gate set targeted by many quantum compilers~\cite{quipper, tket, pyzx}.

\comment{As executing non-Clifford gates in many models of FT computation relies on expensive magic state distillation \cite{knill_2004_faulttolerant, bravyi_2005_magicstate}, a common goal of compilation tools to minimize the number of non-Clifford gates in a transpiled program.}

\begin{figure}[t]
    \centering
    \includegraphics[width=.65\linewidth]{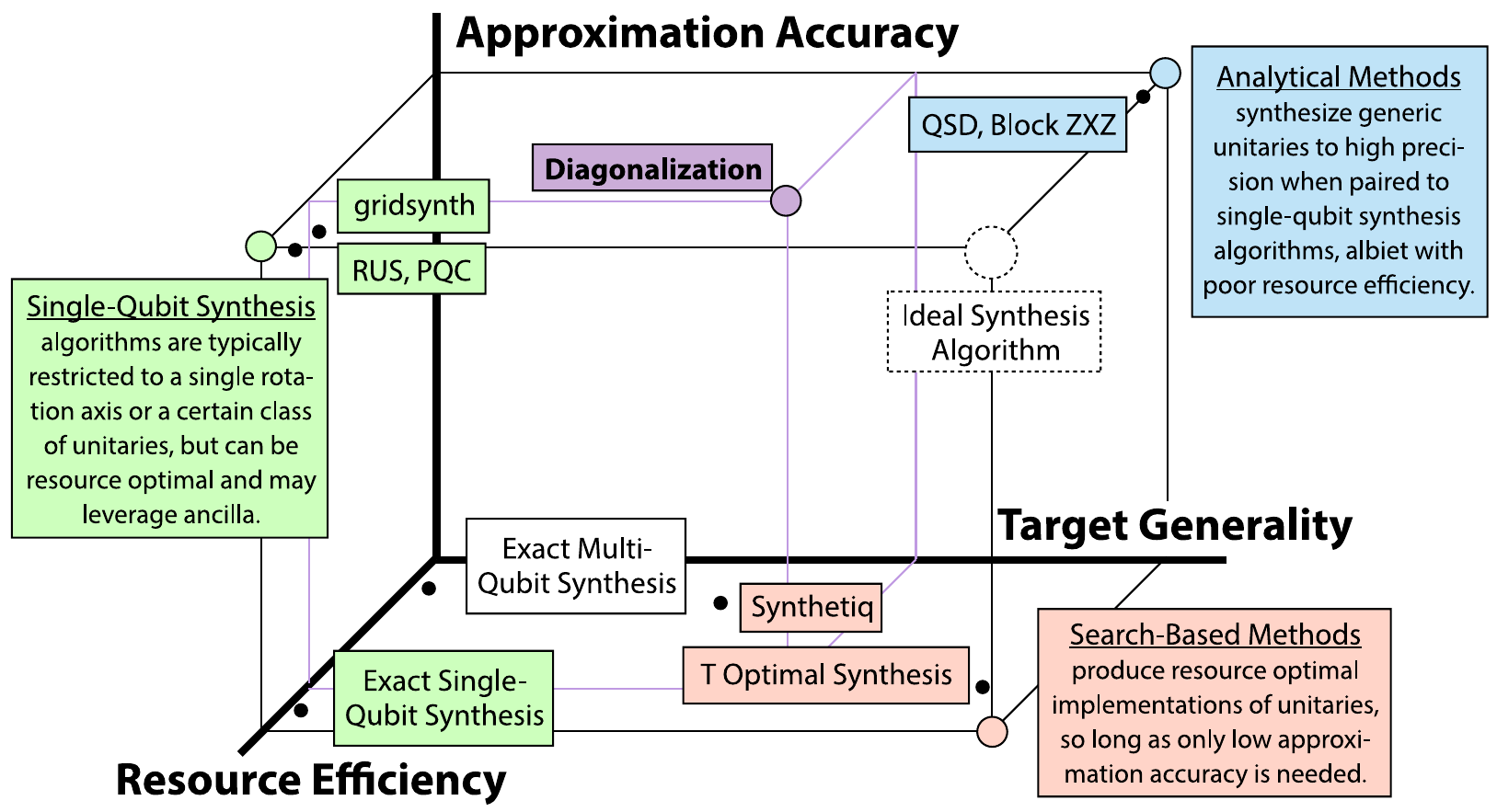} 
    \caption{\footnotesize 
        Tradeoffs for synthesis algorithms targeting the Clifford+T gate set.
        \emph{Approximation accuracy} refers to the precision with which a target program is implemented.
        In the FT setting, this value must be high in order for program outputs to be meaningful.
        \emph{Resource efficiency} refers to non-Clifford gate counts.
        As these gates are orders of magnitude more expensive than Clifford gates, synthesis algorithms should use as few non-Clifford gates as possible.
        \emph{Target generality} refers to the size of the set of unitaries that can be taken as input.
        This ranges from exactly synthesizable unitaries and approximations of 1-qubit $R_Z$ rotations, to the most general class of all multi-qubit unitaries.
    }
    \label{fig:axis}
\end{figure}

Quantum compilers must use approximate unitary synthesis to ensure circuits are transpiled, or translated, into FT gate sets~\cite{solovay_kitaev}.
Synthesis algorithms must balance \textit{resource efficiency} (e.g. non-Clifford gate count), \textit{approximation accuracy} (error) and \textit{target generality} (width and size of set of unitaries that can be taken as input).
Existing approaches can be categorized along the principle \textit{``re\-source efficiency, high precision, and generality: pick any two!"}.

% An algorithm's runtime depends heavily on its resource use, as executing non-Clifford gates in many models of FT computation relies on expensive magic state distillation \cite{knill_2004_faulttolerant, bravyi_2005_magicstate}.

% We present a survey of the state-of-the-art FT synthesis algorithms in
% Section~\ref{sec:backg}. Overall, it appears that there is only one
% generic enough compilation method which employs $R_z$ single qubit
% gate approximate~\cite{gridsynth} synthesis. All other methods are
% either designed for restricted operations on single qubits, or fail to
% generate good circuits for most multi-qubit cases (can't handle
% $R_z$). Consequently, none of the multi-qubit methods have been
% demonstrated on actual large programs.

% Previous work involving synthesis with parameterized (i.e. non-FT) gates demonstrates the benefits of multi-qubit synthesis in resource efficient transpilation~\cite{younis_2022_transpilation}.
% Therefore, the motivation behind this research was to see if we can attain the same benefits in the FT synthesis realm.

% \red{The exponent FT synthesis algorithm tries to build a circuit by placing gates from the ${CNOT, H, S, T}$ gate set one at a
% time.
% In their formulation of synthesis, they attempt to invert the input unitary.
% Functionality is limited by: 1) it may require a large number of gates to implement an unitary; 2) the search over placements has to navigate a huge discrete search space.}

\comment{
The function of any $n$-qubit quantum circuit is fully specified by a unitary matrix $U \in \unitarygroup$.
Unitaries with all elements in the ring $\mathbb{Z}[e^{i\pi/4}, \frac{1}{2}]$ can be implemented exactly in the Clifford+T gate set \cite{kliuchnikov_2012_exact, giles_2013_exactmultiqubit}. However, unitaries which appear in real algorithms rarely follow this form. Fortunately, the Clifford+T gate set can approximate any unitary transformation~\cite{solovay_kitaev}.

General \emph{analytical methods} paired with single-qubit approximate synthesis algorithms can decompose any unitary, but the resultant circuits contain $O(4^N)$ gates~\cite{shende_2006_qsd, devos_2016_blockzxz}.
In restricted cases (i.e., single qubit $R_Z$ rotations), finding optimal Clifford+T approximations of unitaries is efficient \cite{ross_2016_gridsynth}.
Other unitaries, such as multi-qubit diagonal unitaries, can be approximated efficiently, albeit without guarantees of optimality \cite{bullock_2004_diagonal}.
For more general unitaries, an optimal approximation is intractable \cite{gosset_2013_tcount}.

\emph{Search-based} circuit synthesis methods perform discrete searches through the space of possible quantum gate placements \cite{amy_2013_middle, gheorghiu_2022_tcount, paradis_2024_synthetiq}.
These algorithms can discover more efficient circuits than analytical methods, but fail when high precision is needed.
Continuously parameterized rotation operations typically correspond to very long sequences of discrete gates (see Figure~\ref{fig:strategies}).
Existing methods fail in these scenarios because discovering these long sequences is difficult.
}

\comment{To summarize, analytical synthesis algorithms generate high-precision implementations of arbitrary unitaries using a large number of resources, while search-based methods produce very coarse approximations using fewer resources.}

Our work advances FT synthesis along these three axes by combining analytical and search-based methods. 
While existing algorithms attempt to directly invert target unitaries, our main insight is to diagonalize instead: we perform a discrete search until target unitaries are (approximately) diagonalized.
This process reveals single-qubit rotations that are difficult to compile with search-based multi-qubit synthesis algorithms.
Instead of looking for discrete implementations of these continuously param\-eterized gates, we leverage analytical techniques to implement them. Synthesis-by-diagonalization  \linebreak approximates unitaries with precisions that are orders of magnitude higher than other search-based multi-qubit synthesis algorithms without negatively impacting run time and while maintaining good resource efficiency.

% To demonstrate utility we have extended state-of-the-art multi-qubit FT synthesis tools to diagonalize.
% Results indicate significant improvements across all three design criteria: we attain much higher precision and can handle cases where all existing tools fail, all with good resource efficiency.
% To our knowledge, this is the first demonstration of the utility of augmenting  the ubiquitous {\tt gridsynth} $R_z$ synthesis algorithm with multi-qubit synthesis inside a general purpose compiler.
We demonstrate how synthesis-by-diagonalization can implement unitaries that are out of reach for other synthesis algorithms.
We retrofitted a simulated annealing-based synthesis algorithm~\cite{paradis_2024_synthetiq} and trained Reinforcement Learning agents to perform synthesis-by-diagonalization.
This approach is general, and other synthesis tools can also be modified to diagonalize rather than invert.
We also deploy diagonalization in an end-to-end gate set transpilation workflow and demonstrate utility for actual algorithms.
In particular, we observe up to an 18.1\% reduction in T count.
We believe our synthesis tools can help automate the discovery of gadgets that exploit ancilla to improve resource efficiency.
% Futhermore, as discussed in Section~\ref{sec:disc}, the method is extensible and allows for pluging in novel synthesis algorithms, while providing control over approximation accuracy level.
% We believe we show enough motivation to spur their development.

\comment{Diagonalization produces high precision results when state-of-the-art search-based tools fail.
This approach is general; any search-based tool can be retrofitted for diagonalization. % COMMENT: Unlcear that diagonalization is relaxation here, we should bring up something about the degrees of freedom, etc.
Our data indicate that the relaxation to diagonalization enables high precision synthesis of unitaries from real quantum algorithms, and reduces resources used compared to analytical methods.
% Synthesis using diagonalization also reduces T gate counts by up to 16.8\% compared to existing techniques.
}

The remainder of the paper is organized as follows:
Section~\ref{sec:background} discusses necessary background information relating to synthesis with discrete gate sets.
Section~\ref{sec:diagonalization} explains how matrix diagonalization can be framed in the language of synthesis.
Section~\ref{sec:experiments} evaluates diagonalization by synthesizing unitaries taken from partitioned quantum algorithms.
Section~\ref{sec:discussion} discusses methods of scaling diagonalization.
Finally, Section~\ref{sec:conclusion} offers concluding remarks.

\section{Background}
\label{sec:background}
\begin{figure*}[t]
    \centering
    \includegraphics[width=.8\linewidth]{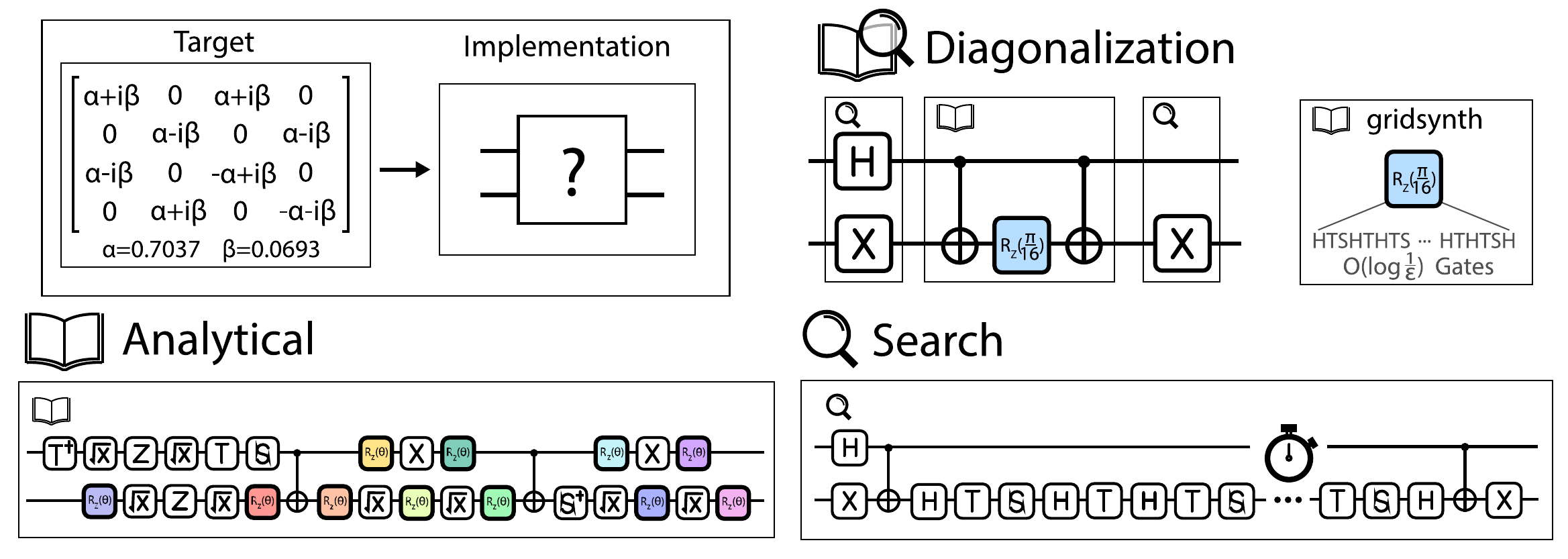}
    \caption{
        Comparison of analytical, search-based, and diagonalization strategies for Clifford+T synthesis.
        Analytical decomposition approaches, such as the Quantum Shannon Decomposition, are universally applicable but use many expensive non-Clifford resources.
        Search-based methods are able to find very low gate count approximations, but are intractable except at low precision.
        Diagonalization is a hybrid approach that enables both high quality and high precision.
        Search is used to diagonalize targets, then analytical methods are used to handle the diagonal results.
        \gridsynth{} \cite{ross_2016_gridsynth} is used to optimally decompose continuous single-qubit $R_Z$ operations (colored boxes) into Clifford+T gates.
    }
    \label{fig:strategies}
\end{figure*}

Fault-Tolerant (FT) quantum computing relies on Quantum Error Correction (QEC) to enable resilient quantum information processing \cite{shor_1995_code}.
Different QEC codes admit different gate sets.
Gates in the Clifford group are classically simulable and often cheap to implement in FT architectures~\cite{knill_2004_faulttolerant, mike_and_ike}.
For universal quantum computation, the logical gate set must also contain a non-Clifford operation~\cite{eastin_2009_eastinknill}.
This work will focus on the approximate compilation of unitaries into the Clifford+T gate set.

The T gate can be executed in a fault-tolerant manner through magic state distillation and injection with Clifford gate corrections \cite{mike_and_ike}. 
This state distillation process is extremely costly.
Estimates show that up to $99\%$ of the resources on a quantum computer could be dedicated to implementing these operations \cite{fowler_2012_surfacecode, O_Gorman_2017}.
Minimizing the number of non-Clifford gates for a given implementation of a unitary is, therefore, an essential step in realizing practical FT quantum computation.

Quantum algorithms are often expressed using operations not directly compatible with QEC codes.
Compilers require methods of transpiling into compatible gate sets.
As multi-qubit unitary synthesis has proven to be a powerful tool for translating between gate sets in the Noisy Intermediate Scale Quantum (NISQ) era~\cite{preskill_2018_nisq}, it is a natural candidate for ensuring algorithms are suitable for FT machines~\cite{younis_2022_transpilation}.
This is a vital component of any FT compiler.

\subsection{Fault-Tolerant Unitary Synthesis}
% In this realm, some constrained unitaries can be exactly implemented.
If every element of an $n$-qubit unitary is in the ring $\mathbb{Z}[e^{i\pi/4}, 1/2]$, it can be implemented exactly with the Clifford+T gate set \cite{kliuchnikov_2012_exact, giles_2013_exactmultiqubit}.
Unitaries appearing in real algorithms are rarely this well structured; they often contain arbitrary angle rotations.

In the single-qubit case, unitaries are decomposed into Clifford $\sqrt{X}$ and non-Clifford $R_Z(\theta)$, so that
\begin{equation}
    U = R_Z(\theta_1) \times  \sqrt{X} \times R_Z(\theta_2) \times \sqrt{X} \times R_Z(\theta_3).
    \label{eqn:single_qubit_unitary}
\end{equation}
These $R_Z$ rotations can be approximated \emph{optimally} to arbitrary precision using the \gridsynth{} algorithm from Ross and Selinger~\cite{ross_2016_gridsynth}.
For some precision $\epsilon$, optimal ancilla-free translation of $R_Z$ gates into Clifford+T operations requires $O(\log 1/\epsilon)$ T gates. This is true regardless of the rotation angle, except in select cases such as $\theta \in \{k\pi/4 : k \in \integers \}$. Ancilla-based methods can generate more resource-efficient approximations~\cite{landahl_2013_cisc, bocharov_2015_rus, bocharov_2015_fallback}.
% Advances in unitary synthesis algorithms promise to discover more such methods, but discovering new techniques requires multi-qubit synthesis.
% COMMENT: Is this true? Should cite something here?
Previous work in circuit compilation demonstrates the utility of multi-qubit unitary synthesis~\cite{davis_2020_qsearch, smith_2021_leap, younis_2022_transpilation, wu_2021_qgo}.
These methods discover approximate circuit implementations using numerical optimization and parameterized (non-FT) gates.
Circuits synthesized this way can be straightforwardly transpiled into the Clifford+T gate set by decomposing parameterized single-qubit gates using Equation~\ref{eqn:single_qubit_unitary} and \gridsynth{}.
While these methods often find circuits with fewer two-qubit gates, they often result in circuits containing many $R_Z$ (and therefore T) gates.

The Quantum Shannon Decomposition and its variants are also powerful synthesis algorithms.
These analytical methods produce circuit implementations that approach the asymptotic lower bound of $O(4^n)$ CNOT and $R_Z$ gates~\cite{shende_2006_qsd, devos_2016_blockzxz}.
Again, this technique results in many non-Clifford gates.

Another option is to deploy a synthesis algorithm that directly operates in the FT gate set and iteratively modifies a circuit, often
gate-by-gate until the target unitary is implemented.
These \emph{search-based} multi-qubit approaches generate more resource-efficient circuits than the previously described generic analytical methods.
Examples of such include simulated annealing~\cite{paradis_2024_synthetiq}, Reinforcement Learning approaches~\cite{zhang_2020_topological_compiling, moro_2021_drlcompiling, chen_2022_efficient, alam_2023_synthesis_mdp, rietsch_2024_unitary}, and several optimal and heuristic synthesis algorithms \cite{amy_2013_middle, gosset_2013_tcount, gheorghiu_2022_tcount}.
However, solving this problem for optimal non-Clifford gate counts is NP Hard~\cite{vandewetering_2024_hard}.

Empirically, these state-of-the-art methods find very efficient implementations of multi-qubit uni\-taries so long as solutions require few (meaning 10s of) gates.
Search-based tools are restricted in that they only operate with discrete gate sets, thus they have no direct way of handling continuous rotations such as $R_Z$ gates.
As a reference, synthesizing a single $R_Z$ gate to a precision of $\epsilon\leq10^{-8}$ requires about $200$ individual Clifford+T gates.
For unitaries containing these rotations, which are ubiquitous in unitaries taken from benchmarks of interest, search-based tools can only find low-precision implementations for a subset of inputs. 

Table~\ref{tab:rlcomparison} compares various unitary synthesis approaches which can be used to target the Clifford+T gate set.
Among search-based methods, the simulated annealing tool \emph{Synthetiq} empirically finds better implementations of unitaries than other methods~\cite{paradis_2024_synthetiq}.
Even so, Synthetiq fails to produce solutions for high precision implementations of complex unitaries as the space of circuits is too large to search.
\emph{Synthesizing complex unitaries to high precision requires a mechanism for handling continuous rot\-ations}; this is possible in both analytical and diagonalization approaches.

Consider the example illustrated in Figure~\ref{fig:strategies}.
In this case, we want to synthesize the unitary labeled \emph{Target}.
Using an analytical synthesis algorithm (e.g., QSD) results in an exponential number of $R_Z$ gates, which are then decomposed into many more T gates. Numerical-optimization methods find much more efficient circuits at the cost of compilation time but still result in far too many T gates. State-of-the-art search-based methods can produce optimal circuits for high-error approximations but have untenable run times as the target precision increases. Ultimately, this lack of precision limits their use in end-to-end compilation (Section \ref{sec:discussion}). Diagonalization combines both the resource-efficiency of search algorithms with the high precision available from analytical methods to practically generate high-quality approximate circuits during the compilation of a wide range of quantum algorithms to an FT gate set.
% The approach taken by the most widely applicable analytical synthesis algorithms (e.g. the QSD) is to decompose the unitary into multi-qubit rotations, then translate these rotations into CNOT and general single-qubit unitary gates. 
% Other numerical-optimization-based synthesis approaches can generate more efficient circuits at the expense of compilation time.
% These methods still produce circuits which contain many non-Clifford gates.
% State-of-the-art search-based methods can produce resource optimal approximate implementations, but have untenable run times.

% Search-based synthesis algorithms require astronomical run times for unitaries such as that in Figure~\ref{fig:strategies}, while purely analytical methods use too many non-Clifford gates.

\begin{table}[t]
    \centering
    \footnotesize
    \begin{tabular}{|c|c|c|c|c|}
        \hline
        Method & Qubits & Category & Precision & Unitary Domain \\
        \hline
        \gridsynth{} \cite{ross_2016_gridsynth} & 1 & Analytical & - & Approx. $R_Z$ \\
        Policy Iter. \cite{alam_2023_synthesis_mdp} & 1 & Search (RL) & $10^{-2}$ & Approx. $\mathbb{U}(2)$ \\
        Synthetiq \cite{paradis_2024_synthetiq} & 1-4 & Search (SA) & $10^{-3}$ & Approx. $\mathbb{U}(2)$-$\mathbb{U}(16)$ \\
        \textbf{Diagonalization} (ours) & 1-3 & Search & $10^{-3}$ and below & Approx. $\cup$ Diagonalizable $\mathbb{U}(2)$-$\mathbb{U}(8)$ \\
        QSD \cite{shende_2006_qsd} & 1+ & Analytical & - & $U \in \unitarygroup$ \\
        \hline
    \end{tabular}
    \caption{
        Unitary synthesis approaches for the Clifford+T gate set.
        Among search-based methods, our \emph{diagonalizing} approach can synthesize targets to higher precision.
        In cases where unitaries can be exactly diagonalized, circuit implementations are produced to arbitrarily high precision.
        Precision here is measured using Hilbert-Schmidt distance (Eq~\ref{eqn:distance}).
        The \emph{Unitary Domain} column indicates what kinds of unitaries can be handled.
        \emph{Synthetiq} uses simulated annealing and empirically outperforms other pure search-based tools at low precision.
        \gridsynth{} optimally decomposes 1-qubit $R_Z$ unitaries to any precision.
        The \emph{QSD} is an analytical method which can be used along with \gridsynth{} to synthesize Clifford+T circuits.
    }
    \label{tab:rlcomparison}
\end{table}

% To our knowledge, the only deployable,  practical and generic FT synthesis method requires a process to generate a circuit with $R_z$ gates and then synthesize their implementation with {\tt gridsynth}. No practical demonstrations of the value of multi-qubit FT synthesis are available in literature. 

\section{Formalizing Unitary Synthesis}
\label{sec:synthesis}
The unitary synthesis is typically framed as finding a circuit which inverts the adjoint (conjugate trans\-pose) of a target matrix.
To solve the problem directly by matrix inversion, every $n$-qubit circuit is described as a sequence of primitive gates taken from a finite set $\actionspace \subset \unitarygroup$.
Every primitive gate has an associated unitary $\ac \in \actionspace$.
A circuit consisting of $t$ gates is associated with a unitary matrix:
\begin{equation}
    \label{eqn:circuitunitary}
    \circuit = \action{t} \times \dots \times \action{1}.
\end{equation}
A circuit represented by $\circuit$ implements a target unitary $\target \in \unitarygroup$ when the distance condition
\begin{equation}
    \label{eqn:distance}
    d_{HS}(\circuit, \target) = \sqrt{1 - \frac{1}{4^n} |Tr(\circuit \target^\dagger)|^2 } \leq \epsilon
\end{equation}
is satisfied.
Here $\epsilon \ll 1$ is a hyperparameter which controls how precisely a target is implemented.

Synthesis methods that invert the adjoint of a target unitary using FT gate sets must find a (likely very long) sequence of discrete gates that satisfies Equation~\ref{eqn:distance}.
We propose an alternative approach: diagonalization (Figures~\ref{fig:strategies} and~\ref{fig:toffoli}).  The diagonalization process consists of a two-headed search that halts when the target's adjoint is diagonalized (i.e., not fully inverted).
Diagonalizing ensures that up to $2^n-1$ continuous $R_Z$ operations can be handled.  
Although there is no guarantee of optimality, this approach is able to find high-precision implementations of complex unitaries where search-based inversion methods fail (Figure~\ref{fig:ccry}), and uses significantly fewer resources than pure analytical methods (Table~\ref{tab:qsd_comparison}).
In the worst case, diagonalization acts exactly like the underlying search-based algorithm which implements it.
When both diagonalization and search-based methods fail, analytical rule-based methods may be necessary.

\begin{figure}[t]
    \centering
    \begin{subfigure}{0.3\textwidth}
        \centering
        \includegraphics[width=0.8\linewidth]{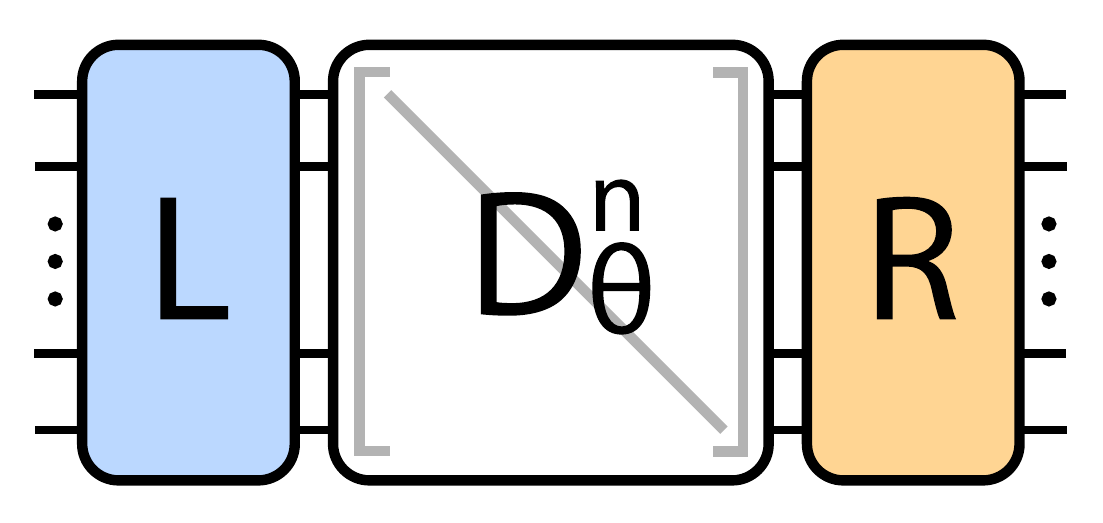}
        \caption{Synthesis-by-Diagonalization Ansatz}
        \label{fig:first}
    \end{subfigure}
    \hfill
    \begin{subfigure}{0.24\textwidth}
        \centering
        \includegraphics[width=\linewidth]{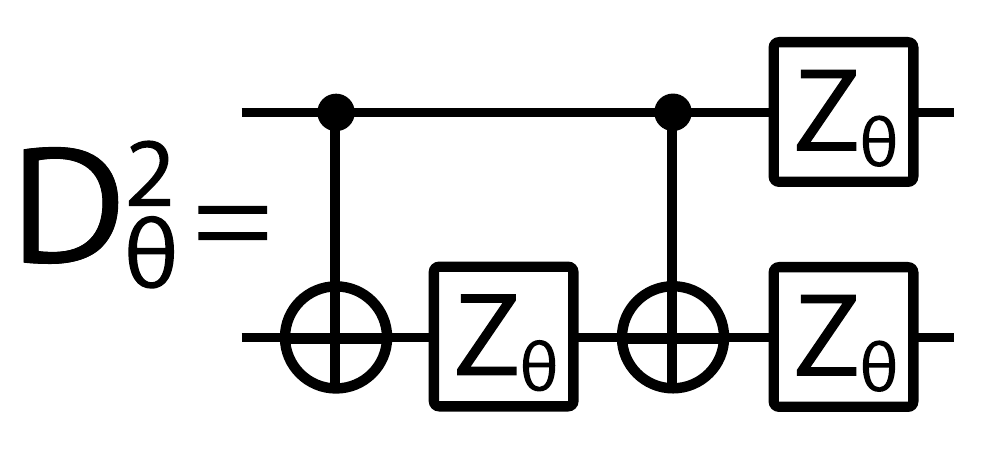}
        \caption{2 Qubit Diagonal Circuit}
        \label{fig:second}
    \end{subfigure}
    \hfill
    \begin{subfigure}{0.33\textwidth}
        \centering
        \includegraphics[width=\linewidth]{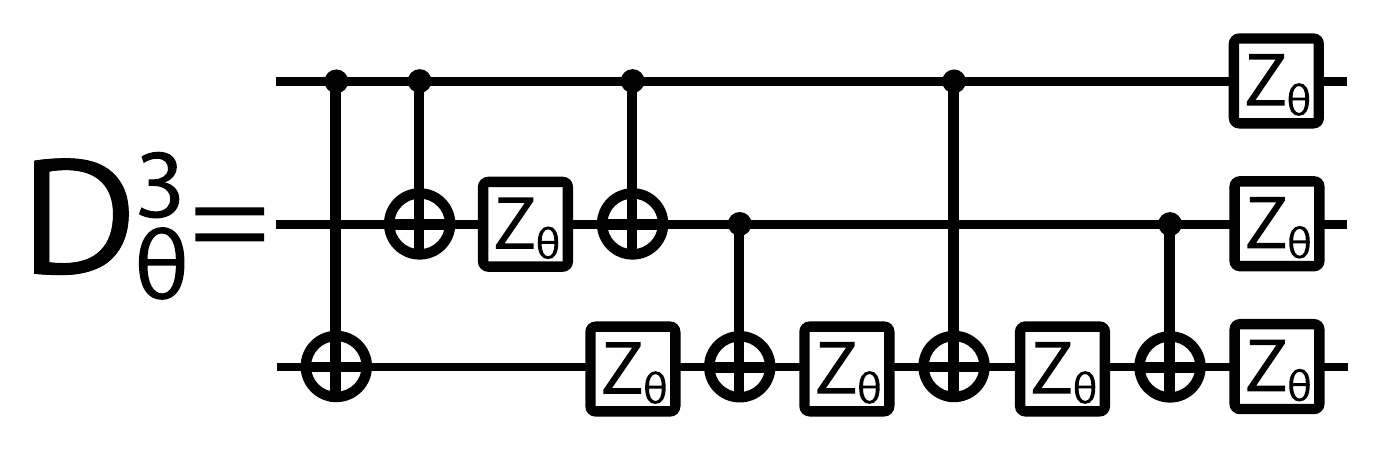}
        \caption{3 Qubit Diagonal Circuit}
        \label{fig:third}
    \end{subfigure}
    \vskip\baselineskip
    \begin{subfigure}{0.65\textwidth}
        \centering
        \includegraphics[width=\linewidth]{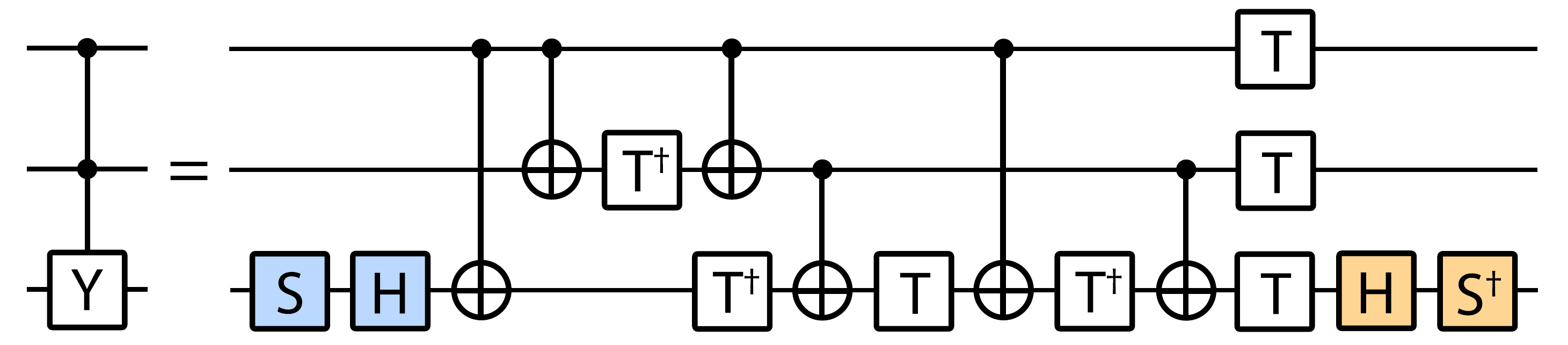}
        \caption{CCY synthesized by diagonalization}
        \label{fig:fourth}
    \end{subfigure}
    \caption{
        Diagonalization ansatz and examples.
        The diagonal matrix $D_\theta$ captures up to $2^n$-$1$ continuous degrees of freedom which we implement with $R_Z(\theta)$ and CNOT gates.
        The $L$ and $R$ subcircuits are made of discrete Clifford+T operations.
        (b-c) Show how to construct 2 and 3 qubit diagonal circuits.
        In many cases, not all the $R_Z$ gates in these circuits are needed.
        (d) A CCY gate only requires two $H$ gates, an $S$, and an $S^\dagger$ to be diagonalized.
        A Toffoli is realized if the $S$ gates are removed.
        Its diagonalization is implemented with only CNOT and $R_Z(\pm\frac{\pi}{4})$ ($T$ or $T^\dagger$ gates).
    }
    \label{fig:toffoli}
\end{figure}

% COMMENT: Does this need it's own paragraph?
\subsection{Synthesis as a Markov Decision Process}
\label{sec:rlsynth}
\label{sec:mdps}
Markov Decision Processes (MDPs) provide a framework for describing stochastic decision problems.
Posing quantum circuit synthesis as an MDP has been demonstrated before \cite{zhang_2020_topological_compiling, moro_2021_drlcompiling, chen_2022_efficient, alam_2023_synthesis_mdp, rietsch_2024_unitary}.
We define the MDP as a tuple $(\statespace, \initialstates, \terminalstates, \actionspace, r)$.
Here $\statespace$ is a set of states, $\initialstates$ a set of initial states, $\terminalstates$ a set of terminal states, $\actionspace$ a set of actions or gates, and $r: \statespace \times \actionspace \rightarrow \realnum$ a reward function indicating when synthesis is done.
We use the term state (or sometimes unitary state) to describe an MDP state $s_t$, not a quantum mechanical state vector or density operator.

\section{Diagonalization}
\label{sec:diagonalization}

Directly synthesizing a unitary by finding a sequence of gates that inverts it is intractable when high precision is needed.
This is because many discrete gates are required to implement a single continuous rotation.
We alleviate the demands placed on search-based synthesis algorithms by using them to diagonalize rather than directly invert unitaries.
% This approach assumes there are continuous rotations in unitaries which cannot be handled by a search-based synthesizer.
Any search-based synthesis algorithm can be used to diagonalize unitaries.
To fit the paradigm of search-based synthesis as described in Section~\ref{sec:synthesis}, we first define the problem of diagonalization as an MDP.
At time $t \in \{0, 1,\dots, T\}$ in the synthesis process, the state $\mdpstate{t} \in \statespace$ is defined by 
\begin{equation}
    \label{eqn:state}
    \mdpstate{t} = L_t \target^\dagger R_t
\end{equation}
where $L_t, R_t$ are each sequences of discrete operations (as in Equation~\ref{eqn:circuitunitary}).
The set of initial states contains the adjoint of target unitaries $S_I = \{ \target^\dagger \}$.
The set of terminal states is the set of all states which can be approximately inverted by a diagonal unitary matrix.
This corresponds to the set of states satisfying
\begin{equation}
    \label{eqn:diagdistance}
    d_{D}(\mdpstate{T}) = d_{HS}\big(\mdpstate{T}, D_\theta\big) \leq \epsilon
\end{equation}
where $D_\theta$ is a diagonal unitary, and $\theta$ is some vector of real rotation angles.
It is sufficient that any $\theta \in \realnum^{2^n-1}$ exists that satisfies this inequality for $\mdpstate{T}$ to be considered a terminal state. Given a terminal state $s_T$, the corresponding circuit is
\begin{equation}
    C_\theta = R_T s_T^\dagger L_T = R_T D_\theta^{-1} L_T.
\end{equation}
Figure~\ref{fig:toffoli} illustrates the general form of three-qubit circuits which satisfy this structure.
As an example, the CCY gate can be diagonalized by just four Clifford gates (2 $H$ gates, an $S$ and an $S^\dagger$ gate).

Synthesis-by-diagonalization can be considered a Singular Value Decomposition, where the left and right singular vector matrices are restricted to unitaries which can be implemented exactly by a discrete gate set (e.g., Clifford+T).
In some cases, diagonalization reduces to inversion (meaning $L_T = D_\theta = I$). % COMMENT: This sounds like you will get the same quality of result (efficiently synthesized) when the graph shows that Synthetiq is better than Diagonalization for low precision. Perhaps this should be clarified
This implies that any circuit which can be efficiently synthesized by inversion can also be synthesized by diagonalization.
\emph{Given the same search algorithm, the set of unitaries that diagonalization is able to synthesize is a strict superset of the unitaries that inversion can synthesize.}

\subsection{The Diagonal Distance}
Our goal is to determine when a state $s_t$ can be nearly inverted.
This means we can satisfy the distance condition in Equation~\ref{eqn:diagdistance} after multiplying by a diagonal unitary.
\begin{theorem}
    A unitary satisfying $\max_{i\in [2^n]} \sqrt{\sum_{j \neq i}|u_{ij}|^2} \leq \epsilon$, where $u_{ij}$ are the unitary's elements, is at most a Hilbert-Schmidt distance of $\epsilon$ away from the identity when multiplied by some diagonal unitary.
\end{theorem}
\begin{proof}
    \noindent 
    Say the state of the synthesis process is $s_t$.
    Each row of $s_t$ is of the form
    \[
        s_t[i] = \begin{bmatrix} u_{i1} & \dots & u_{ii} & \dots & u_{i2^n} \end{bmatrix}.
    \]
    Consider a diagonal unitary matrix $D$ such that for all $i$, $d_{ii} u_{ii} = |u_{ii}|$.
    It then follows that $tr(D \times s_t) = \sum_{i} |u_{ii}|$.
    By the unitarity of $s_t$, we know $|u_{ii}|^2 = 1-\sum_{j\neq i}|u_{ij}|^2$.
    So
    \[
        tr(D \times s_t) = \sum_i |u_{ii}| = \sum_i \sqrt{1 - \sum_{j \neq i} |u_{ij}|^2} \leq \sqrt{4^n (1 - \epsilon^2)}
    \]
    since we assumed that $\sum_{j \neq i} |u_{ij}|^2 \leq \epsilon$.
    This satisfies the Hilbert-Schmidt distance condition (Eq~\ref{eqn:distance}) because
    \[
        d_{HS}(D \times s_t, I) = \sqrt{1 - \frac{1}{4^n}|tr(D \times s_t)|^2} \leq \sqrt{1 - \frac{1}{4^n}|\sqrt{4^n (1-\epsilon^2)}|^2} = \epsilon
    \]
\end{proof}

\subsection{Synthesizing Diagonal Unitaries}
The diagonal operator $D_\theta$ captures continuous degrees of freedom that cannot be efficiently handled by the $L_t$ and $R_t$ unitaries.
We prepare an ansatz implementing $D_\theta$ using specialized algorithms for synthesizing diagonal unitaries \cite{bullock_2004_diagonal}.
Figure~\ref{fig:toffoli} illustrates the circuit structure of generic diagonal unitaries for the 2 and 3 qubit case using only $CNOT$ and $R_Z$ gates.
We chose to implement $R_Z$ rotations in the Clifford+T gate set using \gridsynth{} \cite{ross_2016_gridsynth}, but alternative techniques can also be used.
Some $R_Z$ gates appearing in these ansatzes may not be used in all cases.

\section{Experiments}
\label{sec:experiments}
In this section, we compare synthesis-by-diagonalization with two synthesis-by-inversion algorithms.
The first is Synthetiq~\cite{paradis_2024_synthetiq}, a simulated annealing search-based synthesis algorithm.
The second is the Quantum Shannon Decomposition (QSD)~\cite{shende_2006_qsd}, an analytical synthesis algorithm.

\subsection{Synthesis of Controlled Rotations}
\label{sec:vs_synthetiq}
% Say there are 100 different random angles - synthetiq only handles friendly angles, we won't show anything for it anymore
% Mention success rates for synthetiq
To illustrate the advantages of diagonalization compared to direct inversion with search-based methods, we use both to synthesize controlled rotations.
These primitive gates are ubiquitous in quantum algo\-rithms, including Shor's algorithm~\cite{shor_1997_shorsalgorithm} and Hamiltonian simulation circuits~\cite{li_2022_paulihedral}.
Specifically, we target 100 different random angle $CR_Y(\theta)$ and $CCR_Y(\theta)$ unitaries.
These circuits look similar to Figure~\ref{fig:toffoli}d, but with $R_Z(\theta)$ gates instead of $T$ gates. 
% Comment: I think maybe we should mention that controlled rotations are interesting for certain algorithms?

We compare simulated annealing based diagonalization to direct inversion in Figure~\ref{fig:ccry}.
We modified the simulated annealing synthesis tool, Synthetiq~\cite{paradis_2024_synthetiq}, so that it can perform synthesis by diagonalization.
As mentioned in Section~\ref{sec:diagonalization}, diagonalization is strictly more powerful than inversion; anything which Synthetiq can synthesize by inversion, it can synthesize by diagonalization. 

\enlargethispage{-\baselineskip}

\begin{figure}[bh]
    \centering
    \centering
    \begin{subfigure}{0.49\textwidth}
        \centering
        \includegraphics[width=\linewidth]{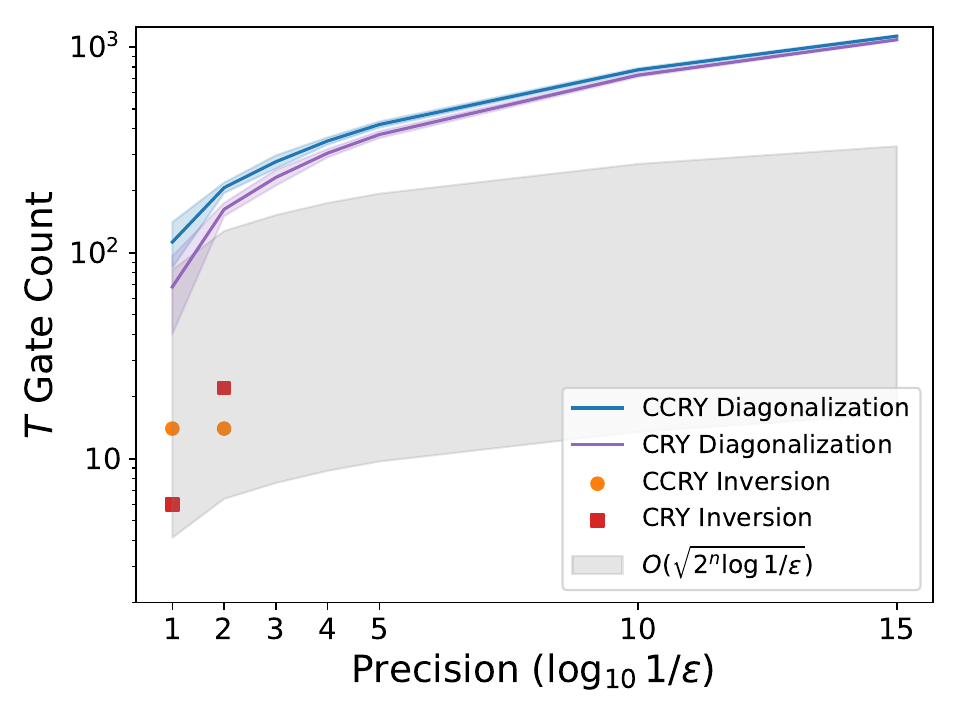}
    \end{subfigure}
    % \hfill
    \begin{subfigure}{0.49\textwidth}
        \centering
        \includegraphics[width=\linewidth]{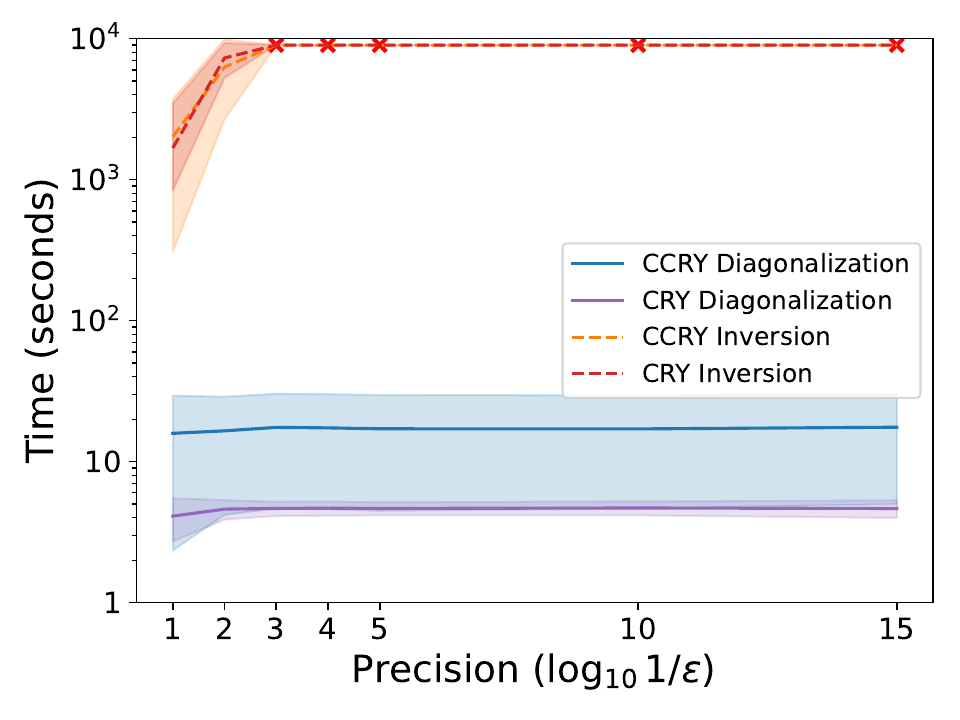}
    \end{subfigure}
    \caption{
        T gate count and run time when synthesizing $CCR_Y$ and $CR_Y$ unitaries with direct-inversion and diagonalization using Synthetiq~\cite{paradis_2024_synthetiq}.
        At low precision, direct-inversion with search-based synthesis can produce optimal results.
        As precision increases beyond $10^{-2}$, diagonalization finds solutions whereas inversion does not.
        The grey area indicates different constant value scalings of the optimal T count for diagonal unitaries~\cite{gosset_2024_diagonaltcount}.
        Improvements in diagonal synthesis algorithms would allow diagonalization to achieve these bounds.
        Diagonalization takes less than 20 seconds across all precisions, inversion times out after 2.5 hours (red ``x" marks indicate time outs).
        The success rate of inversion is 4\% compared to 100\% for diagonalization.
    }
    \label{fig:ccry}
\end{figure}

Low T count implementations of some unitaries can be found for very low precision levels when directly inverting unitaries.
Given a timeout period of 2.5 hours, and running on 64 AMD EPYC 7720P CPU physical cores, the direct inversion method of synthesis is only able to find solutions in 4/100 cases when $\epsilon=10^{-2}$ for both the $CCR_Y$ and $CR_Y$ unitaries.
The diagonalization approach produces implementations that contain more $T$ gates, but its run time does not depend strongly on the target precision, and it finds solutions for all 100 angles tested.
Improvements in synthesis techniques for diagonal operators (see Equation~\ref{eqn:diagtcount}) promise to lower T-counts further.

This experiment highlights how pure search-based methods are largely incapable of synthesizing unitaries which have continuous degrees of freedom.
Search-based synthesis algorithms find resource efficient implementations of unitary matrices at low precisions, but fail when higher precision is needed.

\subsection{Synthesis of Complex Unitaries from Quantum Algorithms}
\label{sec:vs_qsd}

% COMMENT: Fig CCRY is not a more complex unitary, it is also a controlled rotation?
Diagonalization works well for controlled rotations, but also produces high precision implementations of more complex unitaries which appear in quantum algorithms.
In this regime, search-based methods are hopeless (see Figure~\ref{fig:ccry}).
We therefore compare to the Quantum Shannon Decomposition (QSD) \cite{shende_2006_qsd}.

In addition to the simulated annealing diagonalizer mentioned in Section~\ref{sec:vs_synthetiq}, we also trained Re\-inforcement Learning (RL) agents capable of diagonalizing two- and three-qubit unitaries.
We did this because we found inference-based search to be $\approx 1000 \times$ faster than simulated annealing.
Search-based synthesis with RL has been demonstrated before~\cite{zhang_2020_topological_compiling, moro_2021_drlcompiling, chen_2022_efficient, alam_2023_synthesis_mdp, rietsch_2024_unitary}.
To train the diagonalizing agent, we randomly generated separate $A$ and $B$ Clifford+T circuits, then inserted random $D_\theta$ operators to form targets in the form $\target{} = A \times D_\theta \times B$.
Each $A$ and $B$ circuits contained up to 20 random Clifford+T gates.

\begin{table}[!b]
    \tiny
    \fontsize{6pt}{7pt}
    \centering
    \begin{tabular}{|c|c|c|c|c|c|c|c|c|c||c|}
    \hline
    \multirow{13}{*}{\rotatebox{90}{2 Qubits}} & & Heisenberg & HHL & Hubbard & QAOA & QPE & Shor & TFIM & VQE & Mean \\
    \hline
    & Time per & 0.98 & 0.96 & 0.97 & 0.98 & 0.97 & 0.96 & 0.98 & 0.96 & 0.97 \\
    & Unitary (s) & 0.93 & 0.63 & 0.92 & 0.86 & 0.59 & 0.73 & 0.52 & 0.87 & 0.76 \\
    \cline{2-11}
    & Success Rate & 100\% & 100\% & 100\% & 100\% & 100\% & 100\% & 100\% & 100\% & 100\% \\
    & & 93.7\% & 37.4\% & 42.8\% & 27.1\% & 29.1\% & 41.7\%  & 51.0\% & 20.0\% & 42.9\% \\
    \cline{2-11}
    & $R_Z$ Count & 7.07 & 7.58 & 6.34 & 8.27 & 7.88 & 7.31 & 6.15 & 4.88 & 6.94 \\
    & & 1.0 & 1.64 & 0.85 & 1.06 & 0.91 & 1.56 & 1.0 & 0.94 & 1.12 \\
    \cline{2-11}
    & T Count & 0.28 (495.2) & 0.29 (530.9) & 0.04 (443.8) & 0.42 (579.3) & 0.22 (551.8) & 0.15 (511.9) & 0.0 (430.5) & 0.02 (341.6) & 0.18 (486.0) \\
    & & 0.03 (70.0) & 0.36 (115.2) & 0.0 (59.5) & 0.0 (74.2) & 1.15 (64.9) & 0.1 (109.3) & 0.0 (70.0) & 0.16 (66.0) & 0.23 (78.6) \\ 
    \cline{2-11}
    & Clifford Count & 4.66 & 4.20 & 4.63 & 5.15 & 4.35 & 5.77 & 2.06 & 3.33 & 4.27 \\
    & & 4.39 & 4.01 & 3.16 & 2.10 & 3.50 & 4.07 & 5.47 & 8.31 & 4.38 \\
    \cline{2-11}
    & T Reduction & \bf{85.9\%} & \bf{78.4\%} & \bf{86.6\%} & \bf{87.2\%} & \bf{88.5\%} & \bf{78.7\%} & \bf{83.7\%} & \bf{80.7\%} & \bf{83.5\%} \\
    \cline{2-11}
    \hline
    \hline
    \cline{2-11}
    \multirow{12}{*}{\rotatebox{90}{3 Qubits}} & Time per & 1.07 & 1.06 & 1.06 & 1.18 & 1.07 & 1.18 & 1.07 & 1.05 & 1.09 \\
    & Unitary (s) & 1.03 & 1.07 & 1.12 & 1.02 & 1.06 & 1.12 & 1.01 & 0.98 & 1.05 \\
    \cline{2-11}
    & Success Rate & 100\% & 100\% & 100\% & 100\% & 100\% & 100\% & 100\% & 100\% & 100\% \\
    & & 27.4\% & 9.6\% & 55.6\% & 11.2\% & 78.5\% & 18.0\% & 9.4\% & 13.3\% & 27.9\% \\
    \cline{2-11}
    & $R_Z$ Count & 44.60 & 30.42 & 36.73 & 30.96 & 46.89 & 43.67 & 34.63 & 43.57 & 38.93 \\
    & & 1.88 & 2.06 & 0.87 & 2.27 & 1.48 & 4.21 & 1.90 & 0.72 & 1.92 \\
    \cline{2-11}
    & T Count & 1.46 (3124) & 0.68 (2130) & 1.85 (2573) & 1.22 (2168) & 1.26 (3284) & 1.8 (3059) & 1.1 (2425) & 1.84 (3052) & 1.4 (2727) \\
    & & 0.1 (132) & 2.15 (146) & 0.47 (61) & 0.0 (159) & 3.54 (107) & 0.28 (295) & 0.17 (133) & 0.43 (51) & 0.89 (135) \\
    \cline{2-11}
    & Clifford Count & 33.85 & 29.24 & 34.48 & 30.06 & 35.41 & 37.95 & 29.95 & 34.96 & 33.24 \\
    & & 13.06 & 15.94 & 14.85 & 10.22 & 14.00 & 11.92 & 16.76 & 16.69 & 14.18 \\
    \cline{2-11}
    & T Reduction & \bf{95.8\%} & \bf{93.2\%} & \bf{97.6\%} & \bf{92.7\%} & \bf{96.8\%} & \bf{90.4\%} & \bf{94.5\%} & \bf{98.3\%} & \bf{95.1\%} \\
    \cline{2-11}
    \hline
    \end{tabular}
    \caption{
        Synthesis of 2- and 3-qubit unitaries taken from partitioned circuits (see Figure~\ref{fig:transpilation} for an illustration).
        Numbers indicate average time, success rate, Clifford, and non-Clifford gate counts for unitaries taken from a variety of benchmarks.
        To avoid bias, gate counts are only reported for trials where both methods succeed.
        T counts are reported so that the number of T gates introduced by the search based synthesis algorithm appear first, followed by that number plus the number of T gates due to compiling $R_Z$ gates into Clifford+T.
        Compared to the QSD (top rows), diagonalization (bottom rows) on average reduces the number of $R_Z$ gates by 83.5\% (95.1\%) for 2- (3-)qubit unitaries. 
        The average reduction is computed as $(T_{QSD} - T_{Diag})/T_{QSD}$.
        Comparisons are made to the QSD because other synthesis tools fail to find solutions given $\epsilon < 10^{-3}$.
    }
    \label{tab:qsd_comparison}
\end{table}

Throughout these experiments, we require that unitaries be synthesized to a distance of $\epsilon = 10^{-6}$ (see Equation~\ref{eqn:distance}).
Higher precision can be attained in most cases (see Section~\ref{sec:discussion}).
Individual $R_Z$ gates are synthesized using \gridsynth{}.
Each rotation is synthesized to a distance of $\epsilon=10^{-7}$.
In the three-qubit diagonalization case, this means that the total Hilbert-Schmidt distance due to $R_Z$ approximation is at most $7\times10^{-7}$~\cite{wang_1994_trace_inequality}.
QSD results are optimized by simplifying gate sequences and replacing $R_Z$ gates with Clifford+T gates when possible.

We evaluated the diagonalizing agent's performance on a set of unitaries taken from partitioned quantum algorithms (see Figure~\ref{fig:transpilation} for an illustration).
This suite of algorithms includes Shor's Algorithm \cite{shor_1997_shorsalgorithm}, TFIM, Heisenberg, and Hubbard model quantum chemistry simulation circuits \cite{bassman_arqtic_2021}, trained VQE \cite{peruzzo_2014_vqe} and QAOA \cite{farhi_2014_qaoa} circuits, and QPE circuits \cite{kitaev_1995_qpe}.
The VQE and QAOA circuits were generated by MQTBench \cite{quetschlich_2023_mqtbench}.
Each set of unitaries was filtered to ensure that every unitary was unique.
There were 22,323 different two-qubit unitaries and 45,202 different three-qubit unitaries.
Table~\ref{tab:qsd_comparison} summarizes.

We find that in the two-qubit case, diagonalization can find circuits implementing 42.9\% of unitaries across all benchmarks.
About 79\% of these unitaries are already diagonal, and therefore trivial for a diagonalizing agent to synthesize.
Most other unitaries contain 2-8 gates, typically H and S gates.
For the three-qubit case, realistic unitaries are far more complex: they are more likely to contain unitaries which do not conform to the diagonalization ansatz (see Figure~\ref{fig:ccry}a).
Approximately 27.9\% of three-qubit unitaries from our suite of partitioned circuits could be synthesized by the diagonalizing agent.
Of these, approximately 57\% were already diagonal.
On average, the diagonal operators contained $\approx4 \ R_Z$ gates.
The $L_t (\cdot) R_t$ circuits contained an average of about 11 Clifford+T gates.

Compared to the QSD, diagonalization produces solutions with far fewer non-trivial rotation gates.
Implementing an $R_Z$ gate requires $O(\log\frac{1}{\epsilon})$ T gates when synthesized optimally with \gridsynth{} \cite{ross_2016_gridsynth}. 
For $\epsilon=10^{-7}$ this is approximately 70 T gates per $R_Z$.
Lone T gates are almost negligible compared to the contributions from continuous rotations in the high precision regime.
Compared to the QSD, the diagonalizing synthesizer reduces the average number of non-Clifford gates by $83.5\%$ for two qubit unitaries and $95.1\%$ for three qubit unitaries.
As the number of qubits $n$ increases, we expect that diagonalization will outperform the QSD because the former uses at most $2^n-1$ rotations, while the latter typically uses $O(4^n)$.
The advantage of diagonalization compared to the QSD grows as $O(2^n)$.

Although diagonalization does not always succeed, the potential savings in non-Clifford gates and the speed with the process runs remain strong arguments for its utility.
Unitaries which are (nearly) diagonal are ubiquitous primitives in realistic quantum benchmarks.
These common unitaries are well suited for synthesis-by-diagonalization, but entirely out of reach for inversion-based synthesis algorithms.

\begin{figure}[h]
    \centering
    \includegraphics[width=0.55\linewidth]{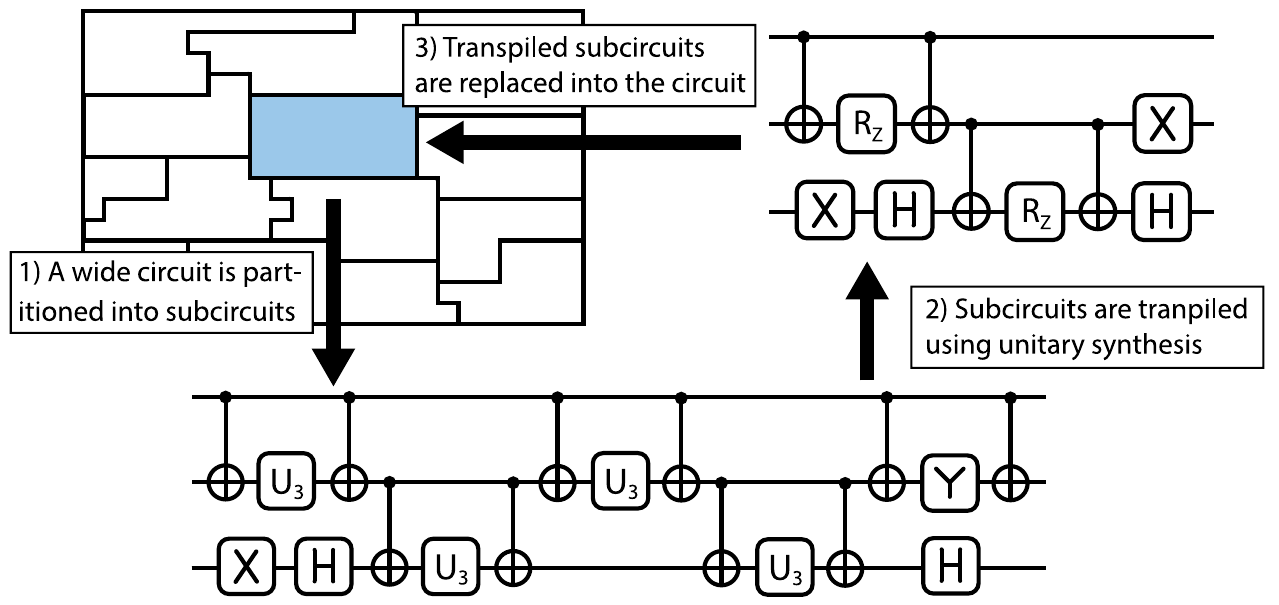}
    \caption{
        Fault-Tolerant gate set transpilation using unitary synthesis. Quantum algorithms are partitioned into many subcircuits. These subcircuits are transpiled independently using unitary synthesis. The optimized and transpiled subcircuits are then replaced into the original circuit. $R_Z$ gates are handled in a post-processing stage using \gridsynth{}.
    }
    \label{fig:transpilation}
\end{figure}

\subsection{End-to-End Circuit Compilation Workflow}
\label{sec:workflow}
Unitary synthesis algorithms are well suited for transpiling quantum circuits into new sets of gates.
In this setting, we assume access to a quantum circuit, not just a unitary matrix.
Past work targeting NISQ gate sets has demonstrated how multi-qubit synthesis can transpile circuits using fewer gates than when using gate-by-gate replacement rules~\cite{younis_2022_transpilation}.
This added optimization power comes from replacing gate-level local translations with more globally aware discovered replacements.
Here we consider whether the ability to handle more complex unitaries with diagonalization yields similar results in FT transpilation.

We consider a \emph{gate-level} transpilation strategy as a control; circuits are transpiled gate-by-gate (as opposed to subcircuit-by-subcircuit) into the Clifford+T gate set.
We compare to this strategy because we already have a gate-level implementation of the algorithm, and the only other method that produces high precision implementations of these unitaries, the QSD, results in an explosion of T gates (Table~\ref{tab:qsd_comparison}).

A summary of the process used is shown in Figure~\ref{fig:transpilation}.
Our method is:
\vspace{-5pt}
\begin{enumerate}[wide, labelindent=1pt]
  \setlength{\itemsep}{-3pt}
  \item Partition a quantum circuit into 2- or 3-qubit blocks that contain as many gates as possible.
  \item For each 2- or 3-qubit unitary, use diagonalization to synthesize a high-precision approximation.
  At the same time, use replacement rules and \gridsynth{} to transpile each partition into Clifford+T gates.
  \item If diagonalization fails or produces results with more T gates, use the gate-level transpilation results.
\end{enumerate}
This process ensures that using multi-qubit synthesis for transpilation never performs worse than the gate-level transpilation strategy.

\renewcommand{\arraystretch}{1.2}
\begin{table}[!b]
    \centering
    \footnotesize
    \begin{tabular}{|c|c|c|c|c|c|c|c|c|}
    \hline
    \multirow{7.5}{*}{\rotatebox{90}{By Gate}} & & Add & HHL & Mult & QAE & QFT & QPE & Shor \\
     & Qubits & 17 & 6 & 16 & 50 & 32 & 30 & 16 \\
    \hline
    \cline{2-9}
    \cline{2-9}
    & $R_Z$ Gates & 251 & 243 & 1,079 & 715 & 1,080 & 321 & 816 \\
    \cline{2-9}
    & $T$ Gates & 23,562 & 22,571 & 100,678 & 65,779 & 101,480 & 29,296 & 74,797 \\
    \cline{2-9}
    % & Cliffords & 35,404 & 34,243 & 152,188 & 97,048 & 152,375 & 45,867 & 117,492 \\
    \cline{2-9}
    & $\epsilon_{\text{total}}$ & $2.5\times 10^{-7}$ & $2.4\times 10^{-7}$ & $1.1\times 10^{-6}$ & $6.7\times 10^{-7}$ & $1.1\times 10^{-6}$ & $3.2\times 10^{-7}$ & $8.2\times 10^{-7}$\\
    \cline{2-9}
    \hline
    \hline
    \cline{2-9}
    & $R_Z$ Gates & 240 & 241 & 1,079 & 683 & 1,080 & 310 & 816 \\
    \cline{2-9}
    \multirow{3}{*}{\rotatebox{90}{2Q Blocks}} & $T$ Gates & 21,924 & 22,201 & 98,516 & 62,405 & 98,640 & 28,290 & 74,790 \\
    \cline{2-9}
    % & Cliffords & 33,867 & 34,194 & 152,188 & 92,675 & 152,268 & 44,348 & 117,492 \\
    \cline{2-9}
    & $\epsilon_{\text{total}}$ & $4.5\times 10^{-7}$ & $2.9\times 10^{-7}$ & $1.2\times 10^{-6}$ & $4.8\times 10^{-6}$ & $4.4\times 10^{-6}$ & $1.1\times 10^{-6}$ & $8.2\times 10^{-7}$\\
    \cline{2-9}
    & \% Diagonalized & 94.0\% & 72.4\% & 92.7\% & 45.1\% & 91.0\% & 21.4\% & 26.2\% \\
    \cline{2-9}
    & Improvement & $7.0\%$ & {1.6\%} & $2.1\%$ & {5.1\%} & $2.8\%$ & {3.4\%} & $0.0\%$ \\
    \cline{2-9}
    \hline
    \hline
    \cline{2-9}
    \cline{2-9}
    & $R_Z$ Gates & 216 & 240 & 898 & 686 & 949 & 316 & 816 \\
    \cline{2-9}
    \multirow{3}{*}{\rotatebox{90}{3Q Blocks}} & $T$ Gates & 19,912 & 22,265 & 82,500 & 62,667 & 87,248 & 28,836 & 74,720 \\
    \cline{2-9}
    % & Cliffords & 30,748 & 33,857 & 127,048 & 96,999 & 135,051 & 45,361 & 117,492 \\
    \cline{2-9}
    & $\epsilon_{\text{total}}$ & $4.2\times 10^{-7}$ & $2.8\times 10^{-7}$ & $1.2\times 10^{-6}$ & $4.2\times 10^{-6}$ & $3.1\times 10^{-6}$ & $9.7\times 10^{-7}$ & $8.2\times 10^{-7}$\\
    \cline{2-9}
    & \% Diagonalized & 85.1\% & 14.3\% & 73.5\% & 40.6\% & 85.3\% & 14.4\% & 12.5\% \\
    \cline{2-9}
    & Improvement & {15.5\%} & $1.4\%$ & {18.1\%} & $4.7\%$ & {14.0\%} & $1.6\%$ & {0.1\%} \\
    \cline{2-9}
    \hline
    \end{tabular}
    \caption{
        Transpilation results. The \emph{By Gate} strategy indicates gate-level transpilation into Clifford+T gates.
        The \emph{2Q} and \emph{3Q Block} strategies indicate circuits partitioned into subcircuits of that size and synthesized by diagonalization.
        In scenarios where diagonalization fails, subcircuits are transpiled gate-by-gate.
        Each subcircuit is synthesized to a distance of $\epsilon=10^{-8}$.
        Individual $R_Z$ gates are synthesized to $\epsilon = 10^{-9}$, resulting in at least 90 $T$ gates per $R_Z$.
        The total approximation error in each circuit on the order of $\epsilon_{\text{total}} \approx 10^{-6}$.
        Reported \emph{Improvement} is percent decrease in T gate counts.
    }
    \label{tab:transpilation_comparison}
\end{table}

Table \ref{tab:transpilation_comparison} summarizes across several quantum algorithms and primitives.
The total approximation error across the entire circuit is about $\epsilon_{\text{total}} \approx 10^{-6}$.
This value is an upper bound that is found by summing the individual approximation errors for each partition and each \gridsynth{} transpiled $R_Z$ gate \cite{wang_1994_trace_inequality}.
Transpilation via diagonalization results in T gate savings for these algorithms compared to gate-level transpilation.
We see the highest T gate count reduction (18.1\%) for the 16 qubit multiplier circuit. 
Other algorithmic primitives such as the 17 qubit adder circuit, and the 32 qubit approximate QFT see similarly large reductions.
The HHL, QAE, and QPE circuits see more moderate decreases in T count.

How well diagonalization transpiles circuits is highly dependent on the circuit partitioning algorithm.
For example, the 16-qubit transpiled implementation of Shor's algorithm sees little improvement com\-pared to the gate-level strategy. 
Using two-qubit blocks results in a success rate of 26\% for this circuit.
In the three-qubit case only 12\% of partitioned subcircuits are successfully diagonalized.
Shor's algorithm consists of many repeated copies QFT and inverse QFT modules, which our data indicate is a class of circuit that can be simplified greatly by diagonalization.
For two- and three-qubit partitions, 91\% and 85\% of subcircuits taken from the example 32 qubit QFT circuit are successfully transpiled by diagonalizing.
The performance of the transpilation strategy is therefore highly likely to be dependent on the partitioning algorithm used.
Partitioning strategies which are informed by circuit structure are likely to improve results.

Combined with circuit partitioning, unitary synthesis enables transpilation to take place on a less localized scale.
This more global view of a circuit's function enables synthesis to outperform simple gate-level methods when paired with approximation in the FT setting.

\section{Discussion}
\label{sec:discussion}
%\subsection{Trade-Offs for High Precision Synthesis of Realistic Unitaries}
FT synthesis algorithms must be able to approximate unitaries to high levels of precision.
In the worst case, the total approximation error in a transpiled circuit is the sum of each individual gate's and part\-ition's approximation error \cite{wang_1994_trace_inequality}.
If the circuits shown in Section~\ref{sec:workflow} (see Table~\ref{tab:transpilation_comparison}) had been transpiled by a synthesis algorithm capable of only finding solutions with precision $\epsilon=10^{-3}$, the average total approximation error of the transpiled circuits would be upper bounded by $\epsilon_{\text{total}} = 0.56$ (ranging from 0.05 to the max error of 1.0).
This much approximation error is unlikely to lead to meaningful algorithm outputs.
The outlook for coarse synthesis is even worse for wider circuits which contain more gates.

Because diagonalization allows for $2^n-1$ continuous rotations to be handled analytically, much higher levels of precision are possible with this approach.
In fact, we can consider the class of \emph{exactly diagonalizable} unitaries (analogous to exactly synthesizable unitaries), where $U = R D_\theta^{-1} L$ and each entry of $L, R \in \integers[e^{i\pi/4}, \frac{1}{2}]$.
These exactly diagonalizable unitaries can be approximated to arbitrarily high levels of precision using synthesis by diagonalization, so long as $L$ and $R$ can be found.
Common examples of these unitaries include controlled rotation gates (Section~\ref{sec:vs_synthetiq}).

Most of the unitaries diagonalization finds solutions for fit into the category of exact diagonalizability.
For this reason, values more precise than $\epsilon=10^{-6}$-$10^{-8}$ (which are shown in this paper) can be attained.
These values are orders of magnitude ($10^3$-$10^5\times$) more precise than previous search-based methods (Table~\ref{tab:rlcomparison}) and are enough to enable high precision transpilation of complete algorithms (Table~\ref{tab:transpilation_comparison}).
Deter\-mining the exact precision needed to ensure circuits produce meaningful results is an open area of research.
How techniques such as unitary mixing~\cite{campbell_2017_mixing} can be used in this setting to boost precision with ensembles of coarse $R_Z$ implementations is worth exploring.
Unitary mixing specifically enables quadratic improvements in precision, meaning this technique can be used to boost precision from $10^{-6}$ to $10^{-12}$.

Our synthesis approach is composable and extensible.
As most FT synthesis work has focused on the Clifford+T gate set, there may be unexplored optimization opportunity when considering alternate gate sets such as Clifford+$\sqrt{T}$ and Clifford+V.
Search-based synthesis tools are well suited to begin answering these questions, as synthesis in alternative gate sets is simply a matter of modifying what gates can be applied.

% We can plug in other algorithms for diagonal synthesis
Recent work~\cite{gosset_2024_diagonaltcount} has proved that a diagonal unitary $D_\theta$ can be approximated to a diamond distance threshold of $\epsilon_\diamond$, where the T count scales as
\begin{equation}
    \text{T-count}(D_\theta) = \Theta(\sqrt{2^n \log \frac{1}{\epsilon_\diamond}} + \log\frac{1}{\epsilon_\diamond}),
    \label{eqn:diagtcount}
\end{equation}
but an efficient algorithm achieving this bound has not been found.
Our approach uses $O(2^n \log\frac{1}{\epsilon})$ T gates to synthesize diagonal unitaries.
Improvements in techniques for synthesizing diagonal unitaries can therefore greatly improve the T-count of our synthesis-by-diagonalization approach.
This improvement in resource efficiency is illustrated by the grey region shown in Figure~\ref{fig:ccry}.

This work showcases the benefit of augmenting analytical decomposition methods with search-based multi-qubit methods. 
% Although diagonalization produces higher quality results than pure analytical decompositions, these results are non-optimal in non-Clifford gate counts.
We expect that other analytical-search-based hybrid approaches that offer further improvements are both possible and practical.

\subsection{The Utility of Diagonalization for Compilation Tasks}
By augmenting search-based compilation with generalized analytical decomposition, diagonalization greatly expands the domain of unitaries which can be transpiled and increases the precision to which they are transpiled.
However, diagonalization also enables even more powerful compilation strategies.

Ancilla qubits and projective measurements exposes more opportunity for optimization when imp\-lementing operations in FT gate sets. 
Repeat Until Success schemes~\cite{bocharov_2015_rus, bocharov_2015_fallback} leverage these resources to reduce the number of non-Clifford gates needed to implement $R_Z$ rotations compared to optimal ancilla-free synthesis.
These techniques rely on synthesizing unitaries with a particular structure (the \emph{Jack of Daggers} structure \cite{bocharov_2015_rus}).
We believe diagonalization and other powerful synthesis techniques will help discover more structures which systematically reduce non-Clifford gate counts when using ancilla.

By loosening the synthesis objective from inversion to diagonalization, our tool exposes opportunity for resource-efficient implementations of continuous rotations. 
As shown in ~\cite{gosset_2024_diagonaltcount} and mentioned in Section~\ref{eqn:diagtcount}, further improvements can be made in the diagonal synthesis procedure.
Improvements here would automatically benefit synthesis-by-diagonalization.

\section{Conclusion}
\label{sec:conclusion}
We have demonstrated a novel approach to high precision multi-qubit unitary synthesis targeting fault-tolerant gate sets.
By diagonalizing unitaries instead of directly inverting them, complex continuous degrees of freedom can be bypassed then handled efficiently with mathematical decomposition methods.
This enables us to synthesize a broader scope of unitaries than other synthesis methods which directly invert continuous rotations, a task which is intractable for high levels of precision.
The diagonalization process is general and other synthesis tools can be retrofitted to diagonalize rather than invert unitaries.

The effectiveness of our diagonalizing approach is demonstrated by synthesizing unitaries taken from a suite of partitioned quantum algorithms to very high precision.
In this regime, only resource inefficient analytical methods are also able to find solutions.
Compared to the Quantum Shannon Decomposition, our approach reduces the number of expensive T gates by $83.5\%$ in two-qubit unitaries and $95.1\%$ in three-qubit unitaries.
Our approach can be used to transpile future-term algorithms to fault-tolerant gate sets with very low approximation error.
Using diagonalization results in up to a 18.1\% reduction in non-Clifford gates compared to gate-by-gate transpilation.
The diagonalizing approach is fast and able to find low gate count implementations of meaningful unitaries, making it a promising technique for use in future fault-tolerant quantum compilers.

%\nocite{*}
\bibliographystyle{eptcs}
\bibliography{refs}

\end{document}